\documentclass[conference, letterpaper]{IEEEtran}
\usepackage{epsfig}
\usepackage{times}
\usepackage{float}
\usepackage{afterpage}
\usepackage{amsmath}
\usepackage{amstext}
\usepackage{amssymb,bm}
\usepackage{latexsym}
\usepackage{color}
\usepackage{graphicx}
\usepackage{amsmath}
\usepackage{amsthm}
\usepackage{graphicx}
\usepackage[center]{caption}
\usepackage{subfigure}
\usepackage{booktabs}
\usepackage{multicol}
\usepackage{lipsum}% dummy text
\usepackage{dblfloatfix}
\usepackage{mathrsfs}
\usepackage{cite}
\usepackage{tikz}
\usepackage{pgfplots}
\pgfplotsset{compat=newest}
 \usepackage{enumitem}
\usepackage{makecell}

\allowdisplaybreaks

\newcommand{\rmd}{\mathrm{d}}
\newcommand{\bbE}{\mathbb{E}}\newcommand{\rme}{\mathrm{e}}

\newcommand{\bbN}{\mathbb{N}}

\newcommand{\bbR}{\mathbb{R}}
\newcommand{\bbS}{\mathbb{S}}

\newcommand{\sfI}{\mathsf{I}}

\newcommand{\sfR}{\mathsf{R}}

\newcommand{\sfX}{\mathsf{X}}

\newcommand{\sfZ}{\mathsf{Z}}

\newcommand{\sfd}{\mathsf{d}}

\newcommand{\cN}{\mathcal{N}}

\newcommand{\sfx}{\mathsf{x}}

\theoremstyle{mystyle}
\newtheorem{theorem}{Theorem}%[section]
\theoremstyle{mystyle}
\newtheorem{lemma}{Lemma}%[section]
\theoremstyle{mystyle}
\newtheorem{prop}{Proposition}%[section]
\theoremstyle{mystyle}
\newtheorem{corollary}{Corollary}%[thm]
\theoremstyle{mystyle}
\theoremstyle{remark}
\newtheorem{rem}{Remark}%[section]
\theoremstyle{mystyle}
\theoremstyle{mystyle}
\theoremstyle{mystyle}
\theoremstyle{discussion}
\theoremstyle{mystyle}
\theoremstyle{mystyle}
 \usepackage{enumitem}

\begin{document}

\title{ Uniform  Distribution on  $(n-1)$-Sphere: Rate-Distortion under Squared Error Distortion}
\author{
 \IEEEauthorblockN{Alex Dytso$^{\dagger}$ and Martina Cardone$^{*}$}
$^{\dagger}$ Qualcomm Flarion Technologies, Bridgewater,  NJ 08807, USA, odytso2@gmail.com\\
$^{*}$ University of Minnesota, Minneapolis, MN 55455, USA, mcardone@umn.edu
}

\maketitle

\begin{abstract}
This paper investigates the rate-distortion function, under a squared error distortion $D$, for an $n$-dimensional random vector uniformly distributed on an $(n-1)$-sphere of radius $R$. 
First, an expression for the rate-distortion function is derived for any values of $n$, $D$, and $R$. 
Second, two types of asymptotics with respect to the rate-distortion function of a Gaussian source are characterized.
More specifically, these asymptotics concern the low-distortion regime (that is, $D \to 0$) and the high-dimensional regime (that is, $n \to \infty$).
\end{abstract}

\section{Introduction}

Consider an $(n-1)$-sphere of radius $R$ defined as 
\begin{equation}
\bbS^{n-1}(R) = \{ \sfx\in \bbR^n : \| \sfx\| =R \},
\end{equation}
where $\|\sfx \|$ is the Euclidean norm, and let $\sfX_R \in \bbR^n$ denote the random vector uniformly distributed on $\bbS^{n-1}(R)$. The random vector $\sfX_R$ appears frequently in various statistical and information theoretic applications, as we summarize next. 

In statistical applications, the distribution of $\sfX_R$ is known to be \emph{the least-favorable distribution}  for the estimation of a bounded normal mean~\cite{BERRY1990130,MARCHAND2002327,fourdrinier2010bayes}. The author of~\cite{BERRY1990130} also provided the expression for the minimum mean squared error (MMSE) of $\sfX_R$.  $\sfX_R$ is also a special case of the \emph{von Mises–Fisher random variable} which has applications in directional statistics~\cite{mardia2000directional}. It is also known that Bayesian estimators with spherically symmetric priors can be written as mixtures of more primitive estimators, namely Bayesian estimators where the prior is the distribution of $\sfX_R$~\cite{robert1990modified}. 

In information theory, the distribution of $\sfX_R$ has several applications. For example, it  is known to be \emph{capacity-achieving} for channels with a peak-power constraint, such as the vector Gaussian channel~\cite{dytso2019capacity} and the vector Gaussian wiretap channel~\cite{wireTapDegraded,vectWireTap}.  
Another application is in the finite blocklength information theory, where it is often used instead of the Gaussian distribution; such applications include point-to-point channels~\cite{tan2015third}, multiple-access channels~\cite{molavianjazi2015second,scarlett2015second}, broadcast channels~\cite{tuninetti2023second}, interference channels~\cite{IC_disp}, and Gel’fand--Pinsker channels\cite{scarlett2015dispersions} to name a few. 

This paper focuses on the rate-distortion function of $\sfX_R$.  The rate-distortion function of $\sfX_R$ has been considered in~\cite{riegler2023lossy}, where a lower bound on it
%the rate-distortion 
has been derived under a variety of distortions, including the squared error distortion.  In contrast, we are interested in characterizing the exact rate-distortion function under the squared error distortion.   The exact rate-distortion function is only known for a handful of sources; for examples, the interested reader is referred to~\cite{berger1998lossy} and references therein.  Distributions on spheres  also appear in spherical quantization~\cite{swaszek1983multidimensional, matschkal2009spherical,eghbali2019deep}, hypersphere learning~\cite{joseph2018adversarial}, and hypothesis testing~\cite{smith1994sphere}.

The distribution of $\sfX_R$ also exhibits several similarities to the Gaussian distribution. For instance, the marginal distribution of the first $k$ components
%dimensions 
of  $\sfX_R$, where $R =\sqrt{n}$, converges to a $k$-dimensional normal distribution in the total variation distance~\cite{stam1982limit} as $n \to \infty$. There are also convergence results of a similar nature for the mutual information and the MMSE.  In particular,  
%it is also known that 
for the Gaussian noise channel, if $\sfX_R$ is used as an input,  then we have the following limits (see Appendix~\ref{sec:MI-MMSE_limits} for the proof): for any $\sigma>0$, it holds that\footnote{For a pair of random vectors $U \in \bbR^{n}$ and $V \in \bbR^{n}$, 
the MMSE is defined as ${\rm mmse}(U|V)= \bbE[ \|U - \bbE[U|V]  \|^2]$.}
\begin{align}
\lim_{n \to \infty: \, R = \sigma \sqrt{n} } \frac{{\rm mmse}(\sfX_{R}| \sfX_R + \sfZ) }{ {\rm mmse}(\sfX_{G}| \sfX_G + \sfZ)} = 1, \label{eq:MMSE_limit} \\ 
\lim_{n \to \infty: \, R = \sigma \sqrt{n} } \frac{I(\sfX_{R}; \sfX_R + \sfZ) }{ I(\sfX_{G}; \sfX_G +\sfZ)} = 1,\label{eq:MI_limit}
\end{align}
where $\sfX_G \sim \cN(0, \sigma^2 I_n)$ and $\sfZ \sim \cN(0,  I_n)$ with $I_n$ being the identity matrix of dimension $n$, and where $\sfX_G$ and $\sfZ$ are independent.

\subsection{Problem Statement }
The goal of this paper is to study the rate-distortion function of $\sfX_R$ under a squared error distortion defined as,
\begin{equation}
\label{eq:RDFGenProbl}
\sfR_n (D;R) =  \!\!\!\inf_{ P_{\hat{\sfX}| \sfX_R }: \hat{\sfX} \in \bbR^n, \, \bbE \left[ \| \hat{\sfX} - \sfX_R \|^2 \right] \le D} I(\hat{\sfX}; \sfX_R) , \, \ D \ge  0. 
\end{equation}
Note that since $\bbE \left [ \| \sfX_R \|^2  \right ]  = R^2$, we assume that $R^2 > D$.
The first objective is to characterize $\sfR_n (D;R)$ in~\eqref{eq:RDFGenProbl} non-asymptotically for every value of $n,D$ and $R$. 

The second objective is to consider two types of asymptotics with respect to the rate-distortion function of a Gaussian source. In other words, we want to understand how the rate-distortion function of $\sfX_R $ compares to the rate-distortion function of a Gaussian random vector.  To make this point clear, recall that  for $\sfX_G \sim \mathcal{N}(0, \sigma^2 I_n)$  the rate-distortion function under a squared error distortion $D$ 
is given by~\cite[Thm.~10.3.3]{cover1999elements}
\begin{equation}
\label{eq:GaussianRDF}
   \sfR_n^{G}(D)= \sfR_n^{G}(D;\sigma^2)  = \frac{n}{2} \log^{+} \left( \frac{n \sigma^2}{D} \right).
\end{equation}
To define the first asymptotic, recall that for a random vector $\sfX$ with rate-distortion function $R_{\sfX}(D)$ the \emph{information dimension} is defined as\footnote{The information dimension is often defined as $\lim_{D \to 0 } \frac{n \sfR_{\sfX}(D)}{\sfR_n^{G}(D)}$. We here choose not to include the multiplicative  $n$ term.  } 
\begin{equation}
\label{eq:InforDim}
\sfd(\sfX)  = \lim_{D \to 0 } \frac{\sfR_{\sfX}(D)}{\sfR_n^{G}(D)}.
\end{equation}
The information dimension measures the rate of growth of the rate-distortion function with respect to the Gaussian rate-distortion function~\cite{kawabata1994rate,wu2010renyi,stotz2016degrees}. 
It is known that  $\sfd(\sfX) = 1$ for random variables whose distributions are absolutely continuous with respect to the Lebesgue measure,   and $\sfd(\sfX) = 0$ for discrete distributions~\cite[Prop.~2]{stotz2016degrees}.  We note that $\sfX_R$ is discrete only for $n=1$, and, for $n>1$, it is singular with respect to the Lebesgue measure and, thus, neither of these results apply.

The second asymptotic that we seek to understand concerns the high-dimensional regime, akin to the limits in~\eqref{eq:MMSE_limit} and~\eqref{eq:MI_limit}. This asymptotic is defined as follows,
\begin{equation}
    \lim_{n  \to \infty} \frac{\sfR_n (D;  {\sqrt{\alpha_n \, n}})}{\sfR_n^{G}(D)},
    \label{eq:LimHighDim}
\end{equation}
where $\alpha_n$ is some function of $n$ (e.g., $\alpha_n =\sqrt{\log n}$).

\subsection{Outline and Contributions}
Section~\ref{sec:preliminaries} focuses on some needed preliminary results pertaining to Bessel functions and some related functions. Section~\ref{sec:main_results} presents our main results.  In particular, Section~\ref{sec:Express_R(D)} presents two  expressions for $\sfR_n (D;R)$ in~\eqref{eq:RDFGenProbl} and it discusses their structures.   Furthermore, it establishes the following Gaussian proximity result,   $\frac{n-1}{n}   \sfR_n^{G}\left(D; 
 \frac{R^2}{n} \right)     \le \sfR_n (D;R) \le  \sfR_n^{G}\left(D; 
 \frac{R^2}{n} \right) $. Section~\ref{sec:low_dist_regime} characterizes the low-distortion limit of $\sfR_n (D;R)$ and it shows that $\sfd(\sfX_R) = 1 -\frac{1}{n}$.  Section~\ref{sec:high_dim_limit} focuses on the high-dimensional behavior and it characterizes the limit in~\eqref{eq:LimHighDim} . For example, it shows that~\eqref{eq:LimHighDim} is equal to one  as long as $\lim_{ n \to \infty} \frac{\log \alpha_n}{\log n} = 0$. 
Section~\ref{sec:MainTheorem} is dedicated to some of the proofs.
The remaining of this section is used to present the notation.

\subsection{Notation}

The modified Bessel function of the first kind of order $\nu$ is denoted by $\sfI_\nu$, and $\Gamma(\cdot)$ is the gamma function. All logarithms are natural.  With $h_b$ we denote the binary entropy.

The surface area of an $(n-1)$-sphere with radius one is denoted as $S_{n-1}$ and given by 
\begin{equation}
\label{eq:SurfaceArea}
S_{n-1} = \frac{2 \pi^{ \frac{n}{2}} }{\Gamma \left( \frac{n}{2} \right)}.
\end{equation}

\section{Preliminaries} 
\label{sec:preliminaries}
Bessel functions and related functions will play an important role in our analysis. Because of this, we next summarize some of these functions and their properties.
%as these will play an important role in our development.  

An approximation of the modified Bessel function that we will use throughout the paper is the following~\cite[eq. 9.6.26]{abramowitz1988handbook},
\begin{equation}
\sfI_{\nu}(t) = \frac{\rme^{t}}{ \sqrt{2 \pi t} } \left(1 - \frac{ 4 \nu^2-1}{ 8 t } + O\left( \frac{1}{t^2} \right) \right). \label{eq:Bessel_function_approx_large-x}
\end{equation}
The following commonly encountered ratio of Bessel functions will play an important role,
\begin{align}
\label{eq:RatioBessel}
f_\nu(t) = \frac{\sfI_\nu(t)}{\sfI_{\nu-1}(t)}, \ t \ge 0. 
\end{align}
The above ratio plays a fundamental role in a variety of application areas~\cite{SeguraBessel}, including information theory~\cite{dytso2019capacity}, signal processing~\cite{schniter2014compressive}, and statistics~\cite{robert1990modified}. In particular, the conditional mean, which is an optimal Bayesian estimator, in Gaussian noise involves the function $f_\nu(t)$~\cite{robert1990modified, dytso2019capacity}.

\begin{lemma} 
\label{lemma:PropertiesRatioBessels}
The function $f_\nu(t)$ in~\eqref{eq:RatioBessel}, with $\nu \geq 1/2$, satisfies the following properties:
\begin{itemize}
    \item $t \mapsto f_{\nu}(t)$ is monotone increasing in $t$~\cite[Thm.~1]{SeguraBessel};
    \item it holds that~\cite[Thm.~1]{SEGURA2011516},~\cite[Thm.~1.1]{Natalini2010}
    %~\cite[Theorem~1]{SeguraBessel}
    \begin{equation}
\label{eq:UpperBoundRatioBessel}
g_\nu(t) \leq f_\nu(t) \leq g_{\nu -\frac{1}{2}}(t) \leq 1,
\end{equation}
where
\begin{equation}
\label{eq:Defgnu}
    g_\nu(t)  = \frac{t}{\nu +\sqrt{\nu ^2+t^2}}.
\end{equation}
    %{\color{blue} ADD other properties that are needed }
\end{itemize}
\end{lemma}

Another two important functions that we will encounter are given by 
\begin{align}
\xi_{\nu}(t)&= - t f_{ \nu } \left(t \right )   + \log \left( (2 \pi)^\nu \frac{\sfI_{\nu-1}\left( t\right )}{ t^{\nu-1}} \right ), \ t  >0, \label{eq:xinu}\\
h_{\nu}(t)& = \xi_{\nu} \left( f_{ \nu }^{-1} \left(t \right ) \right), \ t \in (0,1), \label{eq:hnu}
\end{align}
where we will always have $\nu \geq 1/2$.
The functional inverse $f_{ \nu }^{-1}$ is well defined since $t \mapsto f_{ \nu }(t)$ is monotone increasing for $\nu \geq 1/2$ (see Lemma~\ref{lemma:PropertiesRatioBessels}). The next lemma provides some properties that we will use in the proof of our results. 
\begin{lemma}
\label{lemma:PropertyXinu}
Let $\nu \geq \frac{1}{2}$. Then, we have the following properties: 
\begin{itemize}
\item it holds that
\begin{equation}
f_\nu^{-1}(t)  = 2 \kappa_\nu \frac{t}{1-t^2}, \, t \in (0,1),
\label{eq:InverseOfF}
\end{equation}
where $\nu -\frac{1}{2} \leq \kappa_\nu \leq \nu $;
%$\nu -1\le \kappa_\nu \le \nu -\frac{1}{2} $; {\color{red}MC: Shouldn't this be $\nu -\frac{1}{2} \le \kappa_\nu \le \nu $?}
%
    \item 
$t \mapsto \xi_{\nu}(t)$ is monotone decreasing in $t$;
\item it holds that
\begin{equation}
\lim_{t \to 1^- } h_{\nu}(t)= \left \{  \begin{array}{cc}
%\infty  &  \nu < \frac{1}{2}\\
0  &  \nu = \frac{1}{2}\\
- \infty &  \nu > \frac{1}{2}\\
\end{array} \right. ; \label{eq: h_v limit}
\end{equation}
\item  it holds that 
\begin{equation}
     \lim_{t \to 0^{+} }h_\nu(t) = \log \left(  S_{2\nu-1}\right);  
    \label{eq:hv_to_0}
    \end{equation}
\item for any function $\alpha_\nu $ such that $\lim_{\nu \to  \infty} \alpha_{\nu} \, \nu =\infty $, it holds that
\begin{align}
    \lim_{ \nu \to \infty } \frac {h_{\nu} \left ( \sqrt{1- \frac{1}{ \alpha_{\nu} \nu} }  \right )}{\nu \log \nu }  = & -2 - \lim_{\nu \to \infty}   \frac{ \log  \alpha_\nu }{\log \nu };
    \label{eq:highNuH}
    \end{align}
\item it holds that
\begin{equation}
\label{eq:DerivativeofH}
\frac{{\rm d}}{{\rm d} t} h_\nu(t) = -f_\nu^{-1}(t).
\end{equation}
\end{itemize}
\end{lemma}
\begin{proof}
See Appendix~\ref{app:AuxiliaryProp}.
\end{proof}

\noindent {\bf{Example.}} For $\nu=1/2$, it holds that
    \begin{equation}
f_{\frac{1}{2}}(t)  = \frac{\sfI_{\frac{1}{2}}(t)}{\sfI_{-\frac{1}{2}}(t)}
= \frac{\left( \frac{2}{\pi t} \right )^{\frac{1}{2}}\sinh(t)}{\left( \frac{2}{\pi t} \right )^{\frac{1}{2}}\cosh(t)} = \tanh(t),
\end{equation}
and hence,
\begin{equation}
  f^{-1}_{\frac{1}{2}}(t) =\tanh^{-1}(t) = \frac{1}{2} \log \left (\frac{1+t}{1-t} \right ).  
\end{equation}
Therefore, we arrive at
\begin{align}
\xi_{\frac{1}{2}}(t)&= - t f_{ \frac{1}{2} } \left(t \right )   + \log \left( (2 \pi)^{\frac{1}{2}} \frac{\sfI_{-\frac{1}{2}}\left( t\right )}{ t^{-\frac{1}{2}}} \right )  
\\& = -t \tanh(t) + \log \left( 2 \cosh(t) \right ) 
\\& = -t \frac{\rme^t - \rme^{-t}}{\rme^t + \rme^{-t}} + \log \left( \rme^t + \rme^{-t}\right ),
\end{align}
and hence, 
\begin{equation}
    h_{\frac{1}{2}}(t)=  \frac{1}{2} \log \left( \frac{(1-t)^{t-1}}{(1+t)^{t+1}}\right ) + \log(2) 
 =h_b\left(\frac{1+t}{2} \right),
\label{eq:h1/2}
\end{equation}
where recall that $h_b$ denotes the binary entropy.

\section{Main Results}
\label{sec:main_results}

\subsection{Expressions for $\sfR_n (D;R)$}
\label{sec:Express_R(D)}
The first main result of this paper is provided by the next theorem, the proof of which can be found in Section~\ref{sec:MainTheorem}.
\begin{theorem}
\label{thm:GeneralTheorem}
It holds that
\begin{align}
\label{eq:GeneralExpression}
\sfR_n (D;R) &= \log \left( S_{n-1}\right ) - h_{\frac{n}{2}} \left( \sqrt{1 -\frac{D}{R^2} } \right),
\end{align}
where $S_{n-1}$ is defined in~\eqref{eq:SurfaceArea} and $h_{\nu}(\cdot)$ is defined in~\eqref{eq:hnu}.
\end{theorem}
An immediate consequence of the above theorem is the following corollary.
\begin{corollary}
It holds that
\begin{equation}
\sfR_{1}(D;R) = \log(2) -  h_b \left( \frac{1 + \sqrt{1-\frac{D}{R^2}}}{2} \right ),
\end{equation}
where $h_b$ denotes the binary entropy.
\end{corollary}
We now provide the following lemma, which characterizes the derivative of the rate-distortion function and will be useful 
in a few proofs. 
\begin{lemma} \label{lem:deriaveti_of_RD}
It holds that
   \begin{equation}
\frac{\rmd}{\rmd D}\sfR_n (D;R)  = - \frac{1}{2 R^2 \sqrt{1 - \frac{D}{R^2}}} f_{\frac{n}{2}}^{-1}  \left( \sqrt{1 -\frac{D}{R^2} } \right).
\end{equation}
\end{lemma}
\begin{proof}
From Theorem~\ref{thm:GeneralTheorem}, we have that
\begin{align}
%\label{eq:StartingDerivative}
\frac{\rmd}{\rmd D}\sfR_n (D;R) &= \frac{1}{2 R^2 \sqrt{1 - \frac{D}{R^2}}}  h^\prime_{\frac{n}{2}} \left( \sqrt{1 -\frac{D}{R^2} } \right)
\\ & = - \frac{1}{2 R^2 \sqrt{1 - \frac{D}{R^2}}} f_{\frac{n}{2}}^{-1}  \left( \sqrt{1 -\frac{D}{R^2} } \right),
\end{align}
where the second equality follows from~\eqref{eq:DerivativeofH} in Lemma~\ref{lemma:PropertyXinu}.
This concludes the proof of Lemma~\ref{lem:deriaveti_of_RD}.
\end{proof}
The expression of the rate-distortion function $\sfR_n (D;R)$ in Theorem~\ref{thm:GeneralTheorem} can be rewritten in an integral form, as shown by the next theorem.

\begin{theorem} \label{thm:integral_Express}
    For $ 0 \le  D \le R^2$, it holds that
\begin{equation}
    \sfR_n (D;R) =\int_0^{\sqrt{1 -\frac{D}{R^2} }}    f_{\frac{n}{2}}^{-1}  \left( u \right) \ \rmd u. 
\end{equation}    
\end{theorem}
\begin{proof}
    We have that
    \begin{align}
       & - \sfR_n (D;R) \notag
       \\ & \stackrel{\rm (a)}{=} \sfR_n (R^2;R)  - \sfR_n (D;R) \\
       & \stackrel{\rm (b)}{=} \int_D^{R^2}   - \frac{1}{2 R^2 \sqrt{1 - \frac{t}{R^2}}} f_{\frac{n}{2}}^{-1}  \left( \sqrt{1 -\frac{t}{R^2} } \right) \ \rmd t\\
       & \stackrel{\rm (c)}{=} \int_{\sqrt{1 -\frac{D}{R^2} }}^{0}    f_{\frac{n}{2}}^{-1}  \left( u \right) \rmd u,
    \end{align}
  where the labeled equalities follow from: 
$\rm (a)$   
using the fact that $\lim_{ D \to R^2} \sfR_n (D;R) =0$, which follows from \eqref{eq:hv_to_0}, together with the expression of $\sfR_n (D;R)$ in Theorem~\ref{thm:GeneralTheorem}; 
$\rm (b)$ using Lemma~\ref{lem:deriaveti_of_RD} and the fundamental theorem of calculus; and  $\rm (c)$ applying the change of variable $u = \sqrt{1 -\frac{t}{R^2} }$. This concludes the proof of Theorem~\ref{thm:integral_Express}. 
\end{proof}
Theorem~\ref{thm:integral_Express} is useful to show  several things.  First, it can be used for a numerical implementation of $\sfR_n (D;R)$.  A second application is shown next and it establishes the proximity of $\sfR_n (D;R)$ to the rate-distortion function $\sfR_n^{G}(D;\sigma^2)$ of a Gaussian random vector defined in~\eqref{eq:GaussianRDF}.
\begin{prop} \label{prop:Gaussian_Prox}
    For $ 0 \le  D \le R^2$, it holds that
    \begin{equation}
    \frac{n-1}{n}   \sfR_n^{\text{G}}\left(D; 
 \frac{R^2}{n} \right)      \le  \sfR_n (D;R) \le \sfR_n^{\text{G}}\left(D; 
 \frac{R^2}{n} \right).   
    \end{equation}
\end{prop}   
\begin{proof}
    The proof of the upper bound is a well-known fact about Gaussian random vectors; see for example~\cite[Exercise~10.8]{cover1999elements}. 
To show the lower bound, we start with Theorem~\ref{thm:integral_Express}. We have that
    \begin{align}
    \sfR_n (D;R) & =\int_0^{\sqrt{1 -\frac{D}{R^2} }}    f_{\frac{n}{2}}^{-1}  \left( u \right) \ \rmd u\\
    & \ge  2\left( \frac{n-1}{2}   \right)\int_0^{\sqrt{1 -\frac{D}{R^2} }}    \frac{u}{1-u^2} \ \rmd u\\
    &= \frac{n-1}{n} \sfR_n^{\text{G}}\left(D; 
 \frac{R^2}{n} \right),
    \end{align}
    where the inequality follows from~\eqref{eq:InverseOfF}. This concludes the proof of Proposition~\ref{prop:Gaussian_Prox}. 
\end{proof}

\subsection{Low-Distortion Regime and Information Dimension} 
\label{sec:low_dist_regime}
We here characterize the low distortion limits.   The first limit follows from combining~\eqref{eq:GeneralExpression} and~\eqref{eq: h_v limit}.
\begin{prop}
It holds that
    \begin{equation}
        \lim_{D \to 0^+} \sfR_n (D;R) =  \left\{ \begin{array}{cc}
         \log(2)  & n=1\\
         \infty & n>1
        \end{array} \right.  .
    \end{equation}
\end{prop}
A more refined behavior of the rate-distortion function $\sfR_{\sfX}(D)$ of a random vector $\sfX$  around $D \to 0$ is captured by the information dimension defined in~\eqref{eq:InforDim}. 

The next result provides the information dimension for $\sfX_R$.  We note that 
this has been previously derived in~\cite[eq.~6.32]{riegler2023lossy} by first characterizing the quantization dimension.  Our approach here is more direct since we leverage directly the expression of $\sfR_{n}(D;R)$ in Theorem~\ref{thm:GeneralTheorem}.
\begin{prop}
\label{prop:SmallD}
    Fix  $R >0$ and $n \in \bbN$. Then, it holds that
    \begin{equation}
        \sfd(\sfX_R) = 1-\frac{1}{n}.
    \end{equation}
\end{prop}
\begin{proof}
For $n=1$, we have that $\sfX_R$ is discrete and hence, $\sfd(\sfX_R) =0$~\cite[Prop.~2]{stotz2016degrees}. Therefore, we focus on $n >1.$

We observe the following sequence of steps, 
  \begin{align}
\lim_{D \to 0} \frac{\sfR_{n}(D;R) }{ \log(D) } & \stackrel{\rm (a)}{=}   \lim_{D \to 0} \frac{\frac{\rmd}{\rmd D}\sfR_{n}(D;R) }{ \frac{1}{D} } \\
& \stackrel{\rm (b)}{=}  - \frac{1}{2 R^2 } \lim_{D \to 0}  D  f_{\frac{n}{2}}^{-1}  \left( \sqrt{1 -\frac{D}{R^2} } \right)\\
& \stackrel{\rm (c)}{=}  - \frac{1}{2 R^2 } \lim_{t \to \infty}  R^2 \left(1 - f_{\frac{n}{2}}^2(t)  \right) t\\
& =  - \frac{1}{2  } \lim_{t \to \infty}   \left(1 - f_{\frac{n}{2}}(t)  \right) t  \left(1 +f_{\frac{n}{2}}(t)  \right) \\
& \stackrel{\rm (d)}{=}    \lim_{t \to \infty}   \left( f_{\frac{n}{2}}(t) -1 \right) t  \\
&=   \lim_{t \to \infty}   \left( \frac{ \sfI_{  \frac{n}{2}}(t)-\sfI_{  \frac{n}{2}-1}(t) }{ \sfI_{  \frac{n}{2}-1}(t)}  \right)  t \\
& \stackrel{\rm (e)}{=}      \frac{1}{2}-\frac{n}{2},
\end{align}
where the labeled equalities follow from: 
$\rm (a)$ L'H\^opital's rule;   
$\rm (b)$ Lemma~\ref{lem:deriaveti_of_RD}; 
$\rm (c)$ letting $t = f_{\frac{n}{2}}^{-1}  \left( \sqrt{1 -\frac{D}{R^2} } \right)$ and noting that $D = R^2 \left(1 - f_{\frac{n}{2}}^2(t)  \right) $, and when $D \to 0 $ we have that $\sqrt{1 -\frac{D}{R^2} } \to 1$ and hence, $t \to \infty$ (see~\eqref{eq:InverseOfF}); 
$\rm (d)$  the fact that  $ \lim_{t \to \infty} f_{\frac{n}{2}}(t)  = 1$; 
and $\rm (e)$ the large $t$ approximation in~\eqref{eq:Bessel_function_approx_large-x}. 
The proof of Proposition~\ref{prop:SmallD} is concluded by dividing the above by $-n/2$ as per the definition in~\eqref{eq:InforDim}.
%This concludes the proof of Proposition~\ref{prop:SmallD}.  
\end{proof}

\subsection{ High-Dimensional Regime}
\label{sec:high_dim_limit}
We 
here study
the high dimension behavior of the rate-distortion function $\sfR_n (D;R)$ in~\eqref{eq:GeneralExpression}. 
In particular, we have the following result, which characterizes 
the limit in~\eqref{eq:LimHighDim}.
\begin{prop}
\label{pro:GaussianityPrope}
Consider a function $\alpha_n: \bbN \to \bbR^{+}$ such that $ \lim_{n \to \infty} \alpha_n \, n = \infty$.  Then, it holds that
\begin{equation}
\lim_{n  \to \infty} \frac{\sfR_n (D;  \sqrt{ \alpha_n \, n})}{\sfR_n^{\text{G}}(D)} 
= 1 + \lim_{n \to \infty} \frac{  \log \alpha_n }{\log n}.
\end{equation}
\end{prop}
\begin{proof}
We have that 
\begin{align}
&\lim_{n \to \infty} \frac{\sfR_n (D; \sqrt{\alpha_n \, n})}{\frac{n}{2} \log \left( \frac{n \sigma^2}{D} \right )} \notag
\\ &= \lim_{n \to \infty} \frac{\log \left( S_{n-1}\right ) - h_{\frac{n}{2}} \left( \sqrt{1 -\frac{D}{\alpha_n \,  n} } \right)}{\frac{n}{2} \log \left( \frac{n \sigma^2}{D} \right )}
\\& \stackrel{{\rm{(a)}}}{=} -1 - \lim_{n \to \infty} \frac{ h_{\frac{n}{2}} \left( \sqrt{1 -\frac{D}{\alpha_n \, n} } \right)}{\frac{n}{2} \log \left( \frac{n \sigma^2}{D} \right )}
\\& \stackrel{{\rm{(b)}}}{=} -1 +2 + \lim_{n \to \infty} \frac{ \log \alpha_n }{\log n} ,
\end{align}
where $\rm{(a)}$ follows from the Lanczos approximation and $\rm{(b)}$ is due to~\eqref{eq:highNuH}.
This concludes the proof of Proposition~\ref{pro:GaussianityPrope}.
\end{proof}

From  Proposition~\ref{pro:GaussianityPrope} we note that, as long as the radius does not grow too fast with $n$, i.e.,  $\lim_{n \to \infty} \frac{  \log \alpha_n }{\log n} =0$, the high-dimensional behaviors of $\sfR_n$ and $\sfR_n^G$ are same. 

\begin{rem}
    An alternative way to prove Proposition~\ref{pro:GaussianityPrope} is to rely on the bounds in Proposition~\ref{prop:Gaussian_Prox}. 
\end{rem}

\section{Proof of Theorem~\ref{thm:GeneralTheorem}}
\label{sec:MainTheorem}
In this section, we prove Theorem~\ref{thm:GeneralTheorem}.
We first show that the reconstruction distribution is uniformly supported  on
 $\bbS^{n-1}(r)$ for some $r \geq 0$.
%a shell. 
\begin{lemma}
\label{lemma:RDSimpl1}
It holds that
\begin{align}
\label{eq:FirstStepProof}
\sfR_n (D;R) = \!\!\!\!\!\!\inf_{ P_{\hat{\sfX}_r| \sfX_R }, r \ge 0:   \,  \bbE \left[ \| \hat{\sfX}_r - \sfX_R \|^2 \right] \le D} I(\hat{\sfX}_r; \sfX_R), 
\end{align}
where the marginal of $\hat{\sfX}_r$ is uniformly distributed on  $\bbS^{n-1}(r)$.
     \end{lemma}
\begin{proof}
See Appendix~\ref{sec:ReconstrDistr}.
\end{proof}
The next result reduces the problem in~\eqref{eq:FirstStepProof} 
from optimizing over distributions to a finite dimensional optimization.
\begin{lemma} 
\label{lemma:RDSimpl2}
It holds that
\begin{subequations}
\label{eq:RDSimplified}
\begin{align}
\sfR_{n}(D;R) =  - \max_{r \geq 0} \ \min_{\lambda \geq 0} \ \left( \log \left ( q_{\lambda}(R;r) \right ) + D\lambda \right ),
\end{align}
where
\begin{align}
    q_{\lambda}(R;r) =  
    2^{ \frac{n}{2} -1} 
    \Gamma \left( \frac{n}{2} \right) \rme^{- \lambda ( r^2 +R^2)} \frac{\sfI_{\frac{n}{2} -1}(2\lambda r R) }{ ( 2 \lambda R r )^{\frac{n}{2} -1} }.
    \label{eq:hlambdaGeneral}
\end{align}
\end{subequations}
\end{lemma}
\begin{proof}
We start by noting that
\begin{align}
\sfR_{n}(D;R) &= \!\!\!\!\inf_{ P_{\hat{\sfX}_r| \sfX_R }, r \ge 0:   \,  \bbE \left[ \| \hat{\sfX}_r - \sfX_R \|^2 \right] \le D} I(\hat{\sfX}_r; \sfX_R) \\
&=\!\inf_{r\geq 0} \, \inf_{ P_{\hat{\sfX}_r| \sfX_R }:  \,  \bbE \left[ \| \hat{\sfX}_r \!-\! \sfX_R \|^2 \right] \le D} \!\!I(\hat{\sfX}_r; \sfX_R). \label{eq:Inf_over_r}
\end{align}
We now focus on the inner minimization in the expression above. 
%Starting with the Lagrangian equation 
We have that 
\begin{align}
%\sfR_{n}(D) & \stackrel{{\rm{(a)}}}{=}  
&\inf_{ P_{\hat{\sfX}_r| \sfX_R }:   \bbE \left[ \| \hat{\sfX} - \sfX_R \|^2 \right] \le D} I(\hat{\sfX}_r; \sfX_R) \notag
\\& \stackrel{{\rm{(a)}}}{=} \max_{\lambda \geq 0} \  \inf_{ P_{\hat{\sfX}_r| \sfX_R }} \left ( I(\hat{\sfX}_r; \sfX_R) \right . \notag
\\& \qquad \qquad \qquad \left. + \lambda \left( \bbE \left[ \| \hat{\sfX}_r - \sfX_R \|^2 \right] -D\right ) \right )
\\& \stackrel{{\rm{(b)}}}{=} \max_{\lambda \geq 0} \inf_{ P_{\hat{\sfX}_r| \sfX_R }} \left ( \bbE \left [  \log \frac{\rmd P_{\hat{\sfX}_r|  \sfX_R }}{\rmd P_{\hat{\sfX}_r }} (\hat{\sfX}_r, \sfX_R) \right ] \right .\notag
\\& \qquad \qquad \qquad \left. + \lambda \left( \bbE \left[ \| \hat{\sfX}_r - \sfX_R \|^2 \right] -D\right ) \right )
\\& \stackrel{{\rm{(c)}}}{=}  \max_{\lambda \geq 0}  \ \left( \bbE \left [ \log \frac{{\rm{e}}^{-\lambda \| \hat{\sfX}_r-\sfX_R \|^2}}{{q_{\lambda}(R;r)}}\right ] \right . \notag
\\& \qquad  \qquad \left. + \lambda \left( \bbE \left[ \| \hat{\sfX}_r - \sfX_R \|^2 \right] -D\right ) \right )  \\
&=-  \min_{\lambda \geq 0} \ \left( \log \left ( q_{\lambda}(R;r) \right ) + D\lambda \right ), \label{eq:Last_log_lambda_d_Expreesion}
\end{align}
    where the labeled equalities follow from:
$\rm{(a)}$ using the Lagrange duality theory;
$\rm{(b)}$ using the definition of mutual information;
$\rm{(c)}$ the fact that
\begin{equation}
\label{eq:StepCWithq}
\frac{\rmd P_{\hat{\sfX}_r| \sfX_R } }{\rmd P_{\hat{\sfX}_r}}  (\hat{\sfX}_r, \sfX_R  ) = \frac{{\rm{e}}^{-\lambda \| \hat{\sfX}_r-\sfX_R \|^2}}{\bbE \left [ {\rm{e}}^{-\lambda \| \hat{\sfX}_r-\sfx \|^2}\right ]}= \frac{  {\rm{e}}^{-\lambda \| \hat{\sfX}_r-\sfX_R \|^2}}{q_{\lambda}(R;r)},
\end{equation} 
with $q_{\lambda}(R;r)$ being defined in~\eqref{eq:hlambdaGeneral}
and
where the first equality in~\eqref{eq:StepCWithq} follows from~\cite{cover1999elements} and the second equality is due to~\cite[Prop.~1]{dytso2019capacity}.

Combining~\eqref{eq:Inf_over_r} and~\eqref{eq:Last_log_lambda_d_Expreesion}, we arrive at 
\begin{align}
\sfR_{n}(D;R) =  - \max_{r\geq 0} \ \min_{\lambda \geq 0} \ \left( \log \left ( q_{\lambda}(R;r) \right ) + D\lambda \right ).
\end{align}
This concludes the proof of Lemma~\ref{lemma:RDSimpl2}.
\end{proof}
Now, we solve the optimization problem in~\eqref{eq:RDSimplified}.
In particular, we have the following result, which provides a solution for the inner minimization of the optimization problem in~\eqref{eq:RDSimplified}.
\begin{lemma} 
\label{lemma:ThirdStepSimpl}
It holds that 
\begin{subequations}
\label{eq:ThirdStepOpt}
\begin{align}
\sfR_{n}(D;R) = \log \left( S_{n-1}\right) - \!\!\!\!\!\!\!\max_{r: R - \sqrt{D} \leq r \leq R + \sqrt{D}} h_{\frac{n}{2}} \left ( \delta(r) \right ),
\end{align}
where
\begin{equation}
\delta(r) = \frac{r^2 +R^2 -D}{2 r R},
\end{equation}
\end{subequations}
and where $S_{n-1}$ is defined in~\eqref{eq:SurfaceArea}, and  $h_{\nu}(\cdot)$ is defined in~\eqref{eq:hnu}.
\end{lemma}
\begin{proof}
We define the following function,
\begin{equation}
\label{eq:ulambda}
u(\lambda) = \log \left ( q_{\lambda}(R;r) \right ) + D\lambda,
\end{equation}
and we take its first derivative with respect to $\lambda$. We obtain,
\begin{equation}
u^\prime(\lambda) 
=  2 r R \,f_{ \frac{n}{2} } (2 \lambda r R)+D - r^2 -R^2, \label{eq:DerForLambda}
\end{equation}
where we have used the fact that $ \sfI_\nu^\prime (z) = \frac{\nu}{z} \sfI_\nu (z) + \sfI_{\nu +1}(z)$~\cite[eq. 9.6.26]{abramowitz1988handbook} and where $f_{\nu}(\cdot)$ is defined in~\eqref{eq:RatioBessel}.
Equating~\eqref{eq:DerForLambda} to zero, we arrive at
\begin{equation}
\label{eq:DerivEquLambda}
f_{ \frac{n}{2} } (2 \lambda r R) = \frac{r^2 +R^2 -D}{2 r R} \triangleq \delta(r).
\end{equation}
Note that $\delta(r) > 0$ since $D <R^2$ and $\delta(r) \leq 1$ from~\eqref{eq:UpperBoundRatioBessel} in Lemma~\ref{lemma:PropertiesRatioBessels}.
Solving~\eqref{eq:DerivEquLambda} for $\lambda$, we arrive at
\begin{equation}
\label{eq:LambdaOpt}
\lambda = \frac{1}{2 r R}f_{\frac{n}{2}}^{-1} \left( \delta(r)\right ).
\end{equation}
The above $\lambda$ is indeed the solution of the inner optimization in~\eqref{eq:RDSimplified}. This is because, for $\nu \geq \frac{1}{2}$, $f_\nu(t)$ is monotone increasing in $t$ (see Lemma~\ref{lemma:PropertiesRatioBessels}), which implies that $u(\lambda)$ in~\eqref{eq:ulambda} is convex.
The proof of Lemma~\ref{lemma:ThirdStepSimpl} is concluded by substituting the expression of $\lambda$ in~\eqref{eq:LambdaOpt} inside $\sfR_{n}(D;R)$ in~\eqref{eq:RDSimplified} and by noting that  $\delta(r) \leq 1$ implies $R - \sqrt{D} \leq r \leq R + \sqrt{D}$.
\end{proof}
To complete the proof of Theorem~\ref{thm:GeneralTheorem}, we now solve the optimization in~\eqref{eq:ThirdStepOpt}. We have the following result.
\begin{lemma} 
\label{lemma:ForthStepSimpl}
It holds that 
\begin{align}
\max_{r: R - \sqrt{D} \leq r \leq R + \sqrt{D}} h_{\frac{n}{2}} \left ( \delta(r) \right ) \!=\! \xi_{\frac{n}{2}} \left( f_{\frac{n}{2}}^{-1} \left (\delta \left (\sqrt{R^2\!-\!D} \right ) \right ) \right).
\label{eq:maxh}
\end{align}
\end{lemma}
\begin{proof}
The optimization in~\eqref{eq:ThirdStepOpt} can be rewritten as follows,
\begin{align}
&\max_{r: R - \sqrt{D} \leq r \leq R + \sqrt{D}} h_{\frac{n}{2}} \left ( \delta(r) \right ) \notag
\\& = \max \left \{ \max_{r: R - \sqrt{D} \leq r \leq r^\star} \! h_{\frac{n}{2}} \left ( \delta(r) \right ), \! \!\max_{r: r^\star \leq r \leq R + \sqrt{D}} h_{\frac{n}{2}} \left ( \delta(r) \right ) \right \},
\label{eq:IntermStepLastMax}
\end{align}
where $r^\star= \sqrt{R^2-D}$. 
Now, note that the function $r \to \delta(r)$ is monotone decreasing on $R-\sqrt{D} \le r \le r^\star$ and has a proper inverse on this domain; hence, we have that
\begin{align}
\max_{r : R-\sqrt{D} \le r \leq  r^\star} h_{\frac{n}{2}} \left ( \delta(r) \right ) =  \max_{u :   \delta(r^\star) \leq u \le \delta( R-\sqrt{D} )  }  h_{\frac{n}{2}} \left ( u \right ).
\end{align}
Similarly, note that $ r \to \delta(r)$ is increasing for $r^\star \leq r \le R+\sqrt{D}$ and hence, we have that 
\begin{equation}
\label{eq:qToMaximize}
\max_{r : r^\star \leq   r \le  R +\sqrt{D}} h_{\frac{n}{2}} \left ( \delta(r) \right ) =  \max_{u : \delta(r^\star) \leq  u \le \delta( R +\sqrt{D} )} h_{\frac{n}{2}} \left ( u \right ).
\end{equation}
By using the definition of $h_{\nu}(\cdot)$ in~\eqref{eq:hnu}, we also have that
\begin{align}
& \max_{u :   \delta(r^\star) \leq u \le \delta( R-\sqrt{D} )  }  h_{\frac{n}{2}} \left ( u \right ) 
\\& =  \max_{u :   \delta(r^\star) \leq u \le \delta( R-\sqrt{D} )  } \xi_{\frac{n}{2}} \left( f_{ \frac{n}{2} }^{-1} \left(u \right ) \right) 
\\& \stackrel{{\rm{(a)}}}{=}  \max_{t : 
f_{\frac{n}{2}}^{-1}(\delta(r^\star))  \le   t \le f_{\frac{n}{2}}^{-1}(\delta( R -\sqrt{D} ))  } \xi_{\frac{n}{2}} \left( t \right) 
\\& \stackrel{{\rm{(b)}}}{=} \xi_{\frac{n}{2}} \left( f_{\frac{n}{2}}^{-1}(\delta(r^\star)) \right),
\label{eq:Aux1Lemma6}
\end{align}
where $\rm{(a)}$ follows since $f_{\frac{n}{2}}^{-1}(t)$ is monotone increasing in $t$ (since, from Lemma~\ref{lemma:PropertiesRatioBessels}, $f_{\frac{n}{2}}(t)$ is monotone increasing in $t$) and $\rm{(b)}$ follows from the second property in Lemma~\ref{lemma:PropertyXinu}.
Similarly, we have that
\begin{equation}
 \max_{u :   \delta(r^\star) \leq u \le \delta( R+\sqrt{D} )  }  h_{\frac{n}{2}} \left ( u \right ) = \xi_{\frac{n}{2}} \left( f_{\frac{n}{2}}^{-1}(\delta(r^\star)) \right).
\label{eq:Aux2Lemma6}
\end{equation}
Substituting~\eqref{eq:Aux1Lemma6} and~\eqref{eq:Aux2Lemma6} inside~\eqref{eq:IntermStepLastMax} concludes the proof of Lemma~\ref{lemma:ForthStepSimpl}.
\end{proof}
Now, substituting~\eqref{eq:maxh} inside~\eqref{eq:ThirdStepOpt}, we arrive at
\begin{align}
\sfR_{n}(D;R) &= \log \left( S_{n-1}\right) - \xi_{\frac{n}{2}} \left( f_{\frac{n}{2}}^{-1} \left (\delta \left (\sqrt{R^2\!-\!D} \right ) \right ) \right) 
\\&  = \log \left( S_{n-1}\right) - \xi_{\frac{n}{2}} \left( f_{\frac{n}{2}}^{-1} \left (\sqrt{1 -\frac{D}{R^2} } \right ) \right )
\\& = \log \left( S_{n-1}\right ) - h_{\frac{n}{2}} \left( \sqrt{1 -\frac{D}{R^2} } \right),
\end{align}
which concludes the proof of Theorem~\ref{thm:GeneralTheorem}. 

\appendices

\section{Limits for the MMSE and Mutual Information }
\label{sec:MI-MMSE_limits}

Recall that for $\sfX_G \sim \cN(0,\sigma^2 I_n)$ and $\sfZ \sim \cN(0,  I_n)$ (with $\sfX_G$ and $\sfZ$ being independent), we have that
\begin{align}
I(\sfX_G ; \sfX_G +\sfZ) &= \frac{n}{2} \log \left(1+\sigma^2 \right), \\
{\rm mmse}(\sfX_G | \sfX_G +\sfZ)& = 
 n \frac{\sigma^2}{1 +\sigma^2}.
\end{align}
Now, for $\sfX_R$ 
the MMSE is given by~\cite[Prop.~3]{dytso2019capacity},
\begin{equation}
{\rm mmse}(\sfX_R | \sfX_R +\sfZ) = 
R^2 \left(1 -\bbE \left[ f_{\frac{n}{2}}^2 \left(R \| \sfx +\sfZ \| \right) \right] \right),
 \end{equation}
 where $\sfx \in \bbR^n$ is any vector such that $\|\sfx \|= R$ and $f_{\frac{n}{2}}(\cdot)$ is defined in~\eqref{eq:RatioBessel} with $\nu = \frac{n}{2}$.
Next, by taking $R =\sqrt{\sigma^2 n}$, we have that 
 \begin{align}
 &\lim_{n \to \infty} \frac{{\rm mmse}(\sfX_R | \sfX_R +\sfZ)}{{\rm mmse}(\sfX_G | \sfX_G +\sfZ)}  \notag\\
 &= (1+\sigma^2)  \lim_{n \to \infty} \left(1 -\bbE \left[ f_{\frac{n}{2}}^2 \left(\sqrt{\sigma^2 n} \| \sfx +\sfZ \| \right) \right] \right) \\
 &=  (1+\sigma^2)\left(1 -  \left( \frac{\sigma \sqrt{1+\sigma^2}}{ \frac{1}{2} +\sqrt{ \frac{1}{4} + \sigma^2(1+\sigma^2)  } } \right)^2 \right)  \\
 &=1, 
 \label{eq:mmse_limit}
\end{align}
where the second equality
follows from~\cite[eq.~59]{dytso2019capacity}.

To show the limit for the mutual information, note that for $R =\sqrt{\sigma^2 n}$, we have that
\begin{align}
   &\lim_{n \to \infty} \frac{I(\sfX_R ; \sfX_R +\sfZ)}{n } \notag\\
   & \stackrel{{\rm{(a)}}}{=}  \lim_{n \to \infty}  \int_0^1  \frac{1}{2} \frac{ {\rm mmse}(\sfX_R | \sqrt{\gamma}\sfX_R +\sfZ)}{n}  \ \rmd \gamma \label{eq:Using_i-mmse}\\
   & \stackrel{{\rm{(b)}}}{=}  \frac{1}{2} \int_0^1  \lim_{n \to \infty}  \frac{ {\rm mmse}(\sfX_R | \sqrt{\gamma}\sfX_R +\sfZ)}{n}  \ \rmd \gamma \label{eq:Dominated_Convergence}\\
   & \stackrel{{\rm{(c)}}}{=}   \frac{1}{2} \int_0^1 \frac{1}{\gamma}  \lim_{n \to \infty}  \frac{ {\rm mmse}(\sfX_{\sqrt{\gamma}R} | \sfX_{\sqrt{\gamma} R} +\sfZ)}{n} \ \rmd \gamma \label{eq:rescaling_mmse}\\
   & \stackrel{{\rm{(d)}}}{=}   \frac{1}{2} \int_0^1 \frac{1}{\gamma}  \frac{\gamma \sigma^2 }{1 + \gamma \sigma^2} \ \rmd \gamma \label{eq:using_limit_mmse}\\
   &= \frac{1}{2} \log \left(1+\sigma^2 \right),
\end{align}
where the labeled equalities follow from:
$\rm{(a)}$ using the I-MMSE relationship~\cite{I-mmse};
$\rm{(b)}$ the dominated convergence theorem which is verifiable since 
\begin{equation}
    \frac{{\rm mmse}(\sfX_{R} | \sqrt{\gamma} \sfX_{ R} +\sfZ)}{n} \le  \frac{\bbE[ \|\sfX_R\|^2]}{n} = \sigma^2;
\end{equation}
$\rm{(c)}$ the fact that ${\rm mmse}(\sfX_{R} | \sqrt{\gamma} \sfX_{ R} +\sfZ) =\frac{1}{\gamma } {\rm mmse}(\sfX_{\sqrt{\gamma}R} | \sfX_{\sqrt{\gamma} R} +\sfZ)$, which is a simple consequence of the linearity of expectation; and $\rm{(d)}$ is a consequence of the limit in~\eqref{eq:mmse_limit}.

\section{Proof of Lemma~\ref{lemma:PropertyXinu}}
\label{app:AuxiliaryProp}
\subsection{First Property}
We note that this first property directly follows from~\eqref{eq:UpperBoundRatioBessel}.

\subsection{Second Property}
To show the second property, we take the first derivative of $\xi_{\nu}(t)$ in~\eqref{eq:xinu} with respect to $t$ and we obtain
\begin{align}
\frac{\rmd}{\rmd t} \xi_{\nu}(t) &= -f_{\nu}(t) - t \frac{\rmd}{\rmd t} f_{\nu}(t) + \frac{\frac{\rmd}{\rmd t} \sfI_{\nu-1}(t)}{\sfI_{\nu-1}(t)} -  \frac{\nu-1}{t} 
\\& \stackrel{{\rm{(a)}}}{=} -t + (2 \nu -1) f_{\nu}(t) + t f_{\nu}^2 (t) 
\\& \stackrel{{\rm{(b)}}}{<} - t + \frac{t (2 \nu -1)}{\nu - \frac{1}{2}+\sqrt{\left(\nu-\frac{1}{2} \right )^2+t^2}} \notag 
\\& \quad + \frac{t^3}{\left( \nu - \frac{1}{2}+\sqrt{\left(\nu-\frac{1}{2} \right )^2+t^2} \right )^2}
\\& = 0,
\end{align}
where in the equality in $\rm{(a)}$ we have used the facts that~\cite[eq. 9.6.26]{abramowitz1988handbook}
\begin{subequations}
\label{eq:DerivativeBessels}
\begin{align}
&\frac{\rmd}{\rmd t} \sfI_{\nu}(t) = \sfI_{\nu - 1}(t) - \frac{\nu}{t} \sfI_{\nu}(t),
\\& \frac{\rmd}{\rmd t} \sfI_{\nu -1}(t) = \frac{\nu-1}{t} \sfI_{\nu - 1}(t) +\sfI_{\nu}(t),
\end{align}
\end{subequations}
and the inequality in $\rm{(b)}$ follows from~\eqref{eq:UpperBoundRatioBessel}.

\subsection{Third Property}
To show the third property, we use the large $t$ approximation in~\eqref{eq:Bessel_function_approx_large-x}.
In particular, we have that
\begin{align}
\lim_{t \to 1^- } h_{\nu}(t) & \stackrel{\rm{(a)}}{=}  \lim_{ t \to \infty } \xi_{\nu} \left( t \right)\\
&\stackrel{\rm{(b)}}{=} \lim_{ t \to \infty }  \log \left(  (2 \pi)^\nu \rme^{-t} \frac{\sfI_{\nu-1}\left( t\right )}{ t^{\nu-1}}\right) \\
&\stackrel{\rm{(c)}}{=} \left \{  \begin{array}{cc}
%\infty  &  \nu < \frac{1}{2}\\
0  &  \nu = \frac{1}{2}\\
- \infty &  \nu > \frac{1}{2}\\
\end{array} \right. ,
\end{align} 
where the labeled equalities follow from: 
$\rm{(a)}$ the fact that from~\eqref{eq:InverseOfF} we have that $\lim_{t \to 1^{-} } f_{\nu}^{-1}(t) =\infty $; $\rm{(b)}$ using~\eqref{eq:xinu} and the fact that $\lim_{t \to \infty} f_{\nu}(t)= 1$; and $\rm{(c)}$ applying~\eqref{eq:Bessel_function_approx_large-x}.

\subsection{Fourth Property}
To show the fourth property, we observe that
\begin{align}
    \lim_{t \to 0^{+}}  h_\nu(t) &\stackrel{{\rm{(a)}}}{=} - f_{\nu}^{-1}(t) t  +\log \left(  (2 \pi)^\nu \frac{ \sfI_{\nu-1} (f_{\nu}^{-1}(t))}{ \left( f_{\nu}^{-1}(t) \right) ^{\nu-1}}\right) \\
    & \stackrel{{\rm{(b)}}}{=} \lim_{t \to 0^{+}}  \log \left(  (2 \pi)^\nu \frac{ \sfI_{\nu-1} (t)}{ t^{\nu-1}}\right) \\
    & \stackrel{{\rm{(c)}}}{=} 
    %\lim_{t \to 0^{+}}  
   \log \left(  S_{2\nu-1}\right),
\end{align}
where the labeled equalities follow from:
$\rm{(a)}$ the definition of $h_{\nu}(t)$ in~\eqref{eq:hnu}; 
$\rm{(b)}$ the fact that $\lim_{t \to 0^{+}} f_{\nu}^{-1}(t) =0 $, which can be concluded from~\eqref{eq:InverseOfF}; and $\rm{(c)}$ the limit $\lim_{t \to 0^{+}} \frac{ \sfI_{\nu-1} (t)}{ t^{\nu-1}} = \frac{2^{1-\nu}}{\Gamma \left(\nu \right)}$~\cite[eq. 9.6.7]{abramowitz1988handbook}. 

\subsection{Fifth Property}
To show the fifth property, we observe that
\begin{align}
 &\lim_{ \nu \to \infty } \frac {h_{\nu} \left ( \sqrt{1- \frac{1}{ 
 \alpha_\nu \, \nu} }  \right )}{\nu \log \nu } \notag  \\
 & \stackrel{{\rm{(a)}}}{=} \!\! \lim_{ \nu \to \infty } \!\!\frac{  - f_{\nu}^{-1} \left( \sqrt{1\!-\! \frac{1}{\alpha_\nu \, \nu} } \right)
     \!+\! \log \left(     \sfI_{\nu-1} \left(f_{\nu}^{-1} \left( \sqrt{1\!-\! \frac{1}{\alpha_\nu \, \nu} } \right) \right) \right)}{\nu \log \nu } \notag \\
     & \qquad - \lim_{ \nu \to \infty }  \frac{ \log \left(  f_{\nu}^{-1} \left( \sqrt{1- \frac{1}{\alpha_\nu \, \nu} } \right) \right) }{\log \nu }  \\
     & \stackrel{{\rm{(b)}}}{=}  \lim_{ \nu \to \infty } \frac{  -    2 \alpha_\nu \, \nu^2
     + \log \left(     \sfI_{\nu-1} \left(   2 \alpha_\nu \, \nu^2 \right) \right)}{\nu \log \nu } \notag
     \\& \qquad - \lim_{\nu \to \infty} \frac{ \log  \left( 2 \alpha_\nu \, \nu^2 \right)}{\log \nu } \\
     &=  \lim_{ \nu \to \infty }\!\! \!\frac{ 
     \log \left(   \rme^{ -    2 \alpha_\nu \nu^2}  \sfI_{\nu-1} \left(   2 \alpha_\nu \nu^2\right) \right)}{\nu \log \nu } \!-\!2 \!-\!\!\! \lim_{\nu \to \infty} \!\!  \frac{ \log   \alpha_\nu }{\log \nu },
     \label{eq:FourthPropInterm}
\end{align}
where
$\rm{(a)}$ follows from using the expression of $h_{\nu}(t)$ in~\eqref{eq:hnu} and 
$\rm{(b)}$ is due to~\eqref{eq:InverseOfF}.
Next, the bounds in~\eqref{eq:UpperBoundRatioBessel} lead to
\begin{equation}
I_{0}(t)  \prod_{i=0}^{\nu-1} g_{\nu -\frac{1}{2}-i}(t) > I_\nu(t)  > I_{0}(t)  \prod_{i=0}^{\nu-1} g_{\nu-i}(t).
\end{equation} 
Consequently, we have that
\begin{align}
&\lim_{ \nu \to \infty } \frac{ 
     \log \left(   \rme^{ -    2 \alpha_\nu \, \nu^2}  \sfI_{\nu-1} \left(   2\alpha_\nu \, \nu^2 \right) \right)}{\nu \log \nu } \notag
     \\& > \lim_{ \nu \to \infty } \frac{ 
     \log \left(   \rme^{ -    2\alpha_\nu \, \nu^2}  \sfI_{0} \left(   2\alpha_\nu \, \nu^2 \right) \right)}{\nu \log \nu }  \notag
     \\& \quad + \lim_{ \nu \to \infty } \frac{  \sum_{i=0}^{\nu-2}
     \log \left(    g_{\nu -i-1} \left( 2\alpha_\nu \, \nu^2\right)\right)}{\nu \log \nu }
     \\& \stackrel{{\rm{(c)}}}{=} -\frac{1}{2}\lim_{ \nu \to \infty } \frac{\log(\alpha_\nu)}{\nu \log(\nu)} \nonumber
     \\& \qquad + \lim_{ \nu \to \infty } \frac{  \sum_{i=0}^{\nu-2}
     \log \left(    g_{\nu -i-1} \left( 2\alpha_\nu \, \nu^2 \right)\right)}{\nu \log \nu }
     \\& \stackrel{{\rm{(d)}}}{=} -\frac{1}{2}\lim_{ \nu \to \infty } \frac{\log(\alpha_\nu)}{\nu \log(\nu)} \nonumber
     \\& \qquad +\! \lim_{ \nu \to \infty } \!\frac{  \sum_{i=0}^{\nu-2}
     \log \left(   \frac{2\alpha_\nu \, \nu^2 }{  \nu-i-1 + \sqrt{ ( \nu-i-1)^2 +\left(2\alpha_\nu \, \nu^2 \right)^2 } }   \right)}{\nu \log \nu }
     \\& = -\frac{1}{2}\lim_{ \nu \to \infty } \frac{\log \alpha_\nu}{\nu \log(\nu)} ,
\end{align}
where $\rm{(c)}$ is due to~\eqref{eq:Bessel_function_approx_large-x} and $\rm{(d)}$ follows from~\eqref{eq:Defgnu}.
The upper bound follows along similar lines. 
Combining these facts with~\eqref{eq:FourthPropInterm} leads to
\begin{align}
   & \lim_{ \nu \to \infty } \frac {h_{\nu} \left( \sqrt{1- \frac{1}{ \alpha_\nu \, \nu} }  \right)}{\nu \log \nu }   \notag\\
    &=   -2 - \lim_{\nu \to \infty} \left(    \frac{ \log  \alpha_\nu }{\log \nu } +\frac{1}{2} \frac{\log \alpha_\nu }{\nu \log \nu} \right) \\
    &=- 2 -  \lim_{\nu \to \infty} \frac{ \left( \nu + \frac{1}{2}  \right)\log \alpha_\nu}{\nu \log \nu } \\
    & =- 2 -  \lim_{\nu \to \infty} \frac{\nu + \frac{1}{2}}{\nu} \cdot \lim_{\nu \to \infty} \frac{ \log \alpha_\nu}{\log \nu}\\
    &=- 2 -  \lim_{\nu \to \infty} \frac{ \log \alpha_\nu}{\log \nu}.
    \end{align}

\subsection{Sixth Property}
To show the sixth and last property, we note that from~\eqref{eq:hnu}, we have that
\begin{equation}
\frac{\rmd}{\rmd t} h_\nu(t) =  \xi_{\nu}^\prime \left( f_{ \nu }^{-1} \left(t \right ) \right) \frac{\rmd}{\rmd t} f_{ \nu }^{-1} \left(t \right ) = -f_\nu^{-1}(t),
\end{equation}
where the last equality follows from using~\eqref{eq:DerivativeBessels} and because of the following facts,
\begin{subequations}
\begin{align}
&\frac{\rmd}{\rmd u} \xi(u) = -u + (2 \nu -1) f_{\nu}(u) + u f_{\nu}^2 (u),
\\&\frac{\rmd}{\rmd t} f_{ \nu }^{-1} \left(t \right ) = \frac{1}{ f^\prime_\nu \left(f_\nu^{-1}(t) \right )},
\\&
\frac{\rmd}{\rmd u} f_\nu(u) = 1 - \frac{2\nu-1}{u} f_\nu(u) - f_\nu^2(u).
\end{align}
\end{subequations}
This concludes the proof of Lemma~\ref{lemma:PropertyXinu}.

\section{Proof of Lemma~\ref{lemma:RDSimpl1}}
\label{sec:ReconstrDistr}

First, we note that by using standard Lagrangian duality arguments, we have that
\begin{align}
&\sfR_n (D;R) \notag
\\& =  \inf_{ P_{\hat{\sfX}| \sfX_R }: \, 
\hat{\sfX} \in \bbR^n, \, \bbE \left[ \| \hat{\sfX} - \sfX_R \|^2 \right] \le D} I(\hat{\sfX}; \sfX_R) 
\\& = \max_{\lambda \geq 0} \!\! \inf_{ P_{\hat{\sfX}| \sfX_R }: \, \hat{\sfX} \in \bbR^n } \!\!\!I(\hat{\sfX}; \sfX_R) \!+\! \lambda \left( \bbE \left[ \| \hat{\sfX} \!-\! \sfX_R \|^2 \right] \!-\!D\right ).
\end{align}
We now focus on the inner optimization in the expression above. Specifically, for $\lambda  \geq 0$, we consider 
\begin{equation}
\sfR_{n,\lambda}(D;R) = \inf_{ P_{\hat{\sfX}| \sfX_R }: \, \hat{\sfX} \in \bbR^n } \!\!I(\hat{\sfX}; \sfX_R) \!+\! \lambda \left( \bbE \left[ \| \hat{\sfX} \!-\! \sfX_R \|^2 \right] \!-\!D\right ). \label{eq:lambda_version}
\end{equation}
We leverage the following lemma~\cite{cover1999elements,csiszar1974extremum}, which provides the Karush–Kuhn–Tucker (KKT) conditions for the above optimization problem.
\begin{lemma}
\label{lemm:CondOpti}
Let
\begin{equation}
\label{eq:gHat}
g(\hat{\sfx}) = \bbE \left [ \frac{{\rm{e}}^{-\lambda \| \hat{\sfx}-\sfX_R \|^2}}{q_{\lambda}(\sfX_R)}\right ],
\end{equation}
where
\begin{equation}
\label{eq:hLambda}
q_{\lambda}(\sfx) = \bbE \left [ {\rm{e}}^{-\lambda \| \hat{\sfX}-\sfx \|^2}\right ],
\end{equation}
and where $\lambda \geq 0$.
Then, $P_{\hat{\sfX}}$ is a valid reconstruction distribution in \eqref{eq:lambda_version} if and only if the following holds,
\begin{align}
\begin{array}{ll}
g(\hat{\sfx}) = 1  & {\text{for all}} \ \hat{\sfx} \in S_{\hat{\sfX}} ,
\\ g(\hat{\sfx}) \leq 1  & {\text{for all}} \ \hat{\sfx},
\end{array}
\end{align}
where $S_{\hat{\sfX}}$ is the range of $\hat{\sfX}$.
\end{lemma}
We now make a guess that the reconstruction random vector $\hat{\sfX}$ is uniformly supported on $\bbS^{n-1}(r) $ for some $r\geq 0$, and denote it by $\hat{\sfX}_r$. In this case, the function $q_{\lambda}(\sfx)$ in~\eqref{eq:hLambda} is given by~\cite[Prop.~1]{dytso2019capacity},
\begin{align}
    q_{\lambda}(\sfx) =  
    2^{ \frac{n}{2} -1} 
    \Gamma \left( \frac{n}{2} \right) \rme^{- \lambda ( r^2 +\|\sfx\|^2)} \frac{\sfI_{\frac{n}{2} -1}(2\lambda r \| \sfx \|) }{ ( 2 \lambda \| \sfx \| r )^{\frac{n}{2} -1} }.
    %\label{eq:hlambdaGeneral}
\end{align}
Note that, since $q_{\lambda}(\sfx) $ is only a function of $\| \sfx \|$, we also use the notation
\begin{align}
    q_{\lambda}(\sfx) = q_{\lambda}(\| \sfx \| ;r ), 
\end{align}
where we emphasize the dependence on $r$.

Now, from Lemma~\ref{lemm:CondOpti}, the function $g(\hat{\sfx})$ in~\eqref{eq:gHat} is given~by
\begin{align}
g(\hat{\sfx}) & = \bbE \left [ \frac{{\rm{e}}^{-\lambda \| \hat{\sfx}-\sfX_R \|^2}}{q_{\lambda}(\sfX_R)}\right ] 
\\ & \stackrel{{\rm{(a)}}}{=}\frac{ \bbE \left [ {\rm{e}}^{-\lambda \| \hat{\sfx}-\sfX_R \|^2}\right ]}{q_{\lambda}( R;r)} 
\\& \stackrel{{\rm{(b)}}}{=} \frac{ q_{\lambda}( \| \hat{\sfx} \|;R)}{q_{\lambda}( R;r)} 
\\& \stackrel{{\rm{(c)}}}{=} \frac{ q_{\lambda}( \| \hat{\sfx} \|;R)}{q_{\lambda}( r;R)} ,\label{eq:new_g}
\end{align}
where the labeled equalities follow from:
$\rm{(a)}$ the fact that $q_{\lambda}(\sfX_R)$ depends only on $\|\sfX_R\| = R$ with $R$ being a constant and hence, it can brought outside of the expectation;
$\rm{(b)}$ using~\eqref{eq:hLambda};
and $\rm{(c)}$ the fact that $q_{\lambda}( r;R)= q_{\lambda}(R;r)$

Note that $g(\hat{\sfx})$ is also only a function of $\|\hat{\sfx}\|$ and hence, we can use the notation $g(\hat{\sfx}) = g( \| \hat{\sfx}\|)$.  Combining all of these observations, the conditions in Lemma~\ref{lemm:CondOpti} can be rewritten as,
\begin{align}
g(t) & =1, \, \ t = r,\\
g(t) &\le 1, \, \ t  \in [0,\infty],
\end{align}
or by using~\eqref{eq:new_g} we can rewrite them as,
\begin{align}
 q_{\lambda}( t;R) &= q_{\lambda}( r;R),  \, \ t=r ,\label{eq:trivial_equation}\\
  q_{\lambda}( t;R) &\le q_{\lambda}( r;R), \, \ t \in [0, \infty), 
\end{align} 
where~\eqref{eq:trivial_equation} holds trivially.  
Therefore, to show that a reconstruction distribution is supported on $\bbS^{n-1}(r) $ for some $r\geq 0$,
%a single shell, 
we require to show that there exists an $r\geq 0$ such that 
\begin{align}
  q_{\lambda}( t;R) &\le q_{\lambda}( r;R), \,  \ t \in [0, \infty) .
\end{align} 
Clearly, such an $r$ exists and is given by
\begin{align*}
r^\star_\lambda =\arg \max_{t  \ge 0} \ q_{\lambda}( t;R). 
\end{align*}
At this point, we do not seek to characterize $r^\star_\lambda$, but conclude that for every $\lambda \geq 0$  the minimizing distribution is uniformly supported on
an $(n-1)$-sphere
%a single shell, 
which implies that 
\begin{align}
&\sfR_n (D;R) \notag
\\
& = \inf_{ P_{\hat{\sfX}| \sfX_R }: \hat{\sfX} \in \bbR^n, \,  \bbE \left[ \| \hat{\sfX} - \sfX_R \|^2 \right] \le D} I(\hat{\sfX}; \sfX_R) \\
&=  \inf_{ P_{\hat{\sfX}_r| \sfX_R }, \, r \ge 0:  \,  \bbE \left[ \| \hat{\sfX}_r - \sfX_R \|^2 \right] \le D} I(\hat{\sfX}_r; \sfX_R). 
\end{align}
This concludes the proof of Lemma~\ref{lemma:RDSimpl1}.

 \bibliography{refs.bib}

\begin{thebibliography}{10}
\providecommand{\url}[1]{#1}
\csname url@samestyle\endcsname
\providecommand{\newblock}{\relax}
\providecommand{\bibinfo}[2]{#2}
\providecommand{\BIBentrySTDinterwordspacing}{\spaceskip=0pt\relax}
\providecommand{\BIBentryALTinterwordstretchfactor}{4}
\providecommand{\BIBentryALTinterwordspacing}{\spaceskip=\fontdimen2\font plus
\BIBentryALTinterwordstretchfactor\fontdimen3\font minus
  \fontdimen4\font\relax}
\providecommand{\BIBforeignlanguage}[2]{{%
\expandafter\ifx\csname l@#1\endcsname\relax
\typeout{** WARNING: IEEEtran.bst: No hyphenation pattern has been}%
\typeout{** loaded for the language `#1'. Using the pattern for}%
\typeout{** the default language instead.}%
\else
\language=\csname l@#1\endcsname
\fi
#2}}
\providecommand{\BIBdecl}{\relax}
\BIBdecl

\bibitem{BERRY1990130}
\BIBentryALTinterwordspacing
J.~Berry, ``Minimax {E}stimation of a {B}ounded {N}ormal {M}ean {V}ector,''
  \emph{Journal of Multivariate Analysis}, vol.~35, no.~1, pp. 130--139, 1990.
  [Online]. Available:
  \url{https://www.sciencedirect.com/science/article/pii/0047259X9090020I}
\BIBentrySTDinterwordspacing

\bibitem{MARCHAND2002327}
\BIBentryALTinterwordspacing
{\'E}.~Marchand and F.~Perron, ``On the {M}inimax {E}stimator of a {B}ounded
  {N}ormal {M}ean,'' \emph{Statistics \& Probability Letters}, vol.~58, no.~4,
  pp. 327--333, 2002. [Online]. Available:
  \url{https://www.sciencedirect.com/science/article/pii/S0167715202000895}
\BIBentrySTDinterwordspacing

\bibitem{fourdrinier2010bayes}
D.~Fourdrinier and {\'E}.~Marchand, ``On {B}ayes estimators with uniform priors
  on spheres and their comparative performance with maximum likelihood
  estimators for estimating bounded multivariate normal means,'' \emph{Journal
  of Multivariate Analysis}, vol. 101, no.~6, pp. 1390--1399, 2010.

\bibitem{mardia2000directional}
K.~V. Mardia, P.~E. Jupp, and K.~Mardia, \emph{Directional Statistics}.\hskip
  1em plus 0.5em minus 0.4em\relax Wiley Online Library, 2000, vol.~2.

\bibitem{robert1990modified}
C.~Robert, ``Modified {B}essel functions and their applications in probability
  and statistics,'' \emph{Statistics \& Probability Letters}, vol.~9, no.~2,
  pp. 155--161, 1990.

\bibitem{dytso2019capacity}
A.~Dytso, M.~Al, H.~V. Poor, and S.~S. Shitz, ``On the {C}apacity of the {P}eak
  {P}ower {C}onstrained {V}ector {G}aussian {C}hannel: An {E}stimation
  {T}heoretic {P}erspective,'' \emph{IEEE Transactions on Information Theory},
  vol.~65, no.~6, pp. 3907--3921, 2019.

\bibitem{wireTapDegraded}
A.~Dytso, M.~Egan, S.~M. Perlaza, H.~V. Poor, and S.~S. Shitz, ``Optimal
  {I}nputs for {S}ome {C}lasses of {D}egraded {W}iretap {C}hannels,'' in
  \emph{2018 IEEE Information Theory Workshop (ITW)}, 2018, pp. 1--5.

\bibitem{vectWireTap}
\BIBentryALTinterwordspacing
A.~Favano, L.~Barletta, and A.~Dytso, ``Amplitude {C}onstrained {V}ector
  {G}aussian {W}iretap {C}hannel: {P}roperties of the
  {S}ecrecy-{C}apacity-{A}chieving {I}nput {D}istribution,'' \emph{Entropy},
  vol.~25, no.~5, 2023. [Online]. Available:
  \url{https://www.mdpi.com/1099-4300/25/5/741}
\BIBentrySTDinterwordspacing

\bibitem{tan2015third}
V.~Y.~F. Tan and M.~Tomamichel, ``The {T}hird-{O}rder {T}erm in the {N}ormal
  {A}pproximation for the {AWGN} {C}hannel,'' \emph{IEEE Transactions on
  Information Theory}, vol.~61, no.~5, pp. 2430--2438, 2015.

\bibitem{molavianjazi2015second}
E.~MolavianJazi and J.~N. Laneman, ``A {S}econd-{O}rder {A}chievable {R}ate
  {R}egion for {G}aussian {M}ulti-{A}ccess {C}hannels via a {C}entral {L}imit
  {T}heorem for {F}unctions,'' \emph{IEEE Transactions on Information Theory},
  vol.~61, no.~12, pp. 6719--6733, 2015.

\bibitem{scarlett2015second}
J.~Scarlett and V.~Y. Tan, ``Second-{O}rder {A}symptotics for the {G}aussian
  {MAC} with {D}egraded {M}essage {S}ets,'' \emph{IEEE Transactions on
  Information Theory}, vol.~61, no.~12, pp. 6700--6718, 2015.

\bibitem{tuninetti2023second}
D.~Tuninetti, P.~Sheldon, B.~Smida, and N.~Devroye, ``On {S}econd {O}rder
  {R}ate {R}egions for the {S}tatic {S}calar {G}aussian {B}roadcast
  {C}hannel,'' \emph{IEEE Journal on Selected Areas in Communications},
  vol.~41, no.~7, pp. 1982--1999, 2023.

\bibitem{IC_disp}
S.-Q. Le, V.~Y.~F. Tan, and M.~Motani, ``A {C}ase {W}here {I}nterference {D}oes
  {N}ot {A}ffect the {C}hannel {D}ispersion,'' \emph{IEEE Transactions on
  Information Theory}, vol.~61, no.~5, pp. 2439--2453, 2015.

\bibitem{scarlett2015dispersions}
J.~Scarlett, ``On the {D}ispersions of the {G}el’fand--{P}insker {C}hannel
  and {D}irty {P}aper {C}oding,'' \emph{IEEE Transactions on Information
  Theory}, vol.~61, no.~9, pp. 4569--4586, 2015.

\bibitem{riegler2023lossy}
E.~Riegler, G.~Koliander, and H.~B{\"o}lcskei, ``Lossy compression of general
  random variables,'' \emph{Information and Inference: A Journal of the IMA},
  vol.~12, no.~3, p. 1759–1829, 2023.

\bibitem{berger1998lossy}
T.~Berger and J.~D. Gibson, ``Lossy {S}ource {C}oding,'' \emph{IEEE
  Transactions on Information Theory}, vol.~44, no.~6, pp. 2693--2723, 1998.

\bibitem{swaszek1983multidimensional}
P.~Swaszek and J.~Thomas, ``Multidimensional {S}pherical {C}oordinates
  {Q}uantization,'' \emph{IEEE Transactions on Information Theory}, vol.~29,
  no.~4, pp. 570--576, 1983.

\bibitem{matschkal2009spherical}
B.~Matschkal and J.~B. Huber, ``Spherical {L}ogarithmic {Q}uantization,''
  \emph{IEEE transactions on audio, speech, and language processing}, vol.~18,
  no.~1, pp. 126--140, 2009.

\bibitem{eghbali2019deep}
S.~Eghbali and L.~Tahvildari, ``Deep {S}pherical {Q}uantization for {I}mage
  {S}earch,'' in \emph{Proceedings of the IEEE/CVF Conference on Computer
  Vision and Pattern Recognition}, 2019, pp. 11\,690--11\,699.

\bibitem{joseph2018adversarial}
A.~D. Joseph, B.~Nelson, B.~I. Rubinstein, and J.~Tygar, \emph{Adversarial
  {M}achine {L}earning}.\hskip 1em plus 0.5em minus 0.4em\relax Cambridge
  University Press, 2018.

\bibitem{smith1994sphere}
D.~P. Smith, ``The {S}phere as a {T}ool for {T}eaching {S}tatistics,''
  \emph{Journal of Geological Education}, vol.~42, no.~5, pp. 412--416, 1994.

\bibitem{stam1982limit}
A.~J. Stam, ``Limit theorems for uniform distributions on spheres in
  high-dimensional {E}uclidean spaces,'' \emph{Journal of Applied probability},
  vol.~19, no.~1, pp. 221--228, 1982.

\bibitem{cover1999elements}
T.~Cover and J.~Thomas, \emph{Elements of Information Theory: Second
  Edition}.\hskip 1em plus 0.5em minus 0.4em\relax Wiley, 2006.

\bibitem{kawabata1994rate}
T.~Kawabata and A.~Dembo, ``The {R}ate-{D}istortion {D}imension of {S}ets and
  {M}easures,'' \emph{IEEE Transactions on Information Theory}, vol.~40, no.~5,
  pp. 1564--1572, 1994.

\bibitem{wu2010renyi}
Y.~Wu and S.~Verd{\'u}, ``R{\'e}nyi {I}nformation {D}imension: {F}undamental
  {L}imits of {A}lmost {L}ossless {A}nalog {C}ompression,'' \emph{IEEE
  Transactions on Information Theory}, vol.~56, no.~8, pp. 3721--3748, 2010.

\bibitem{stotz2016degrees}
D.~Stotz and H.~B{\"o}lcskei, ``Degrees of {F}reedom in {V}ector {I}nterference
  {C}hannels,'' \emph{IEEE Transactions on Information Theory}, vol.~62, no.~7,
  pp. 4172--4197, 2016.

\bibitem{abramowitz1988handbook}
M.~Abramowitz, I.~A. Stegun, and R.~H. Romer, \emph{Handbook of Mathematical
  Functions with Formulas, Graphs, and Mathematical Tables}.\hskip 1em plus
  0.5em minus 0.4em\relax American Association of Physics Teachers, 1988.

\bibitem{SeguraBessel}
\BIBentryALTinterwordspacing
J.~Segura, ``Monotonicity {P}roperties for {R}atios and {P}roducts of
  {M}odified {B}essel {F}unctions and {S}harp {T}rigonometric {B}ounds,''
  \emph{Results Math}, vol.~74, no. 221, 2021. [Online]. Available:
  \url{https://doi.org/10.1007/s00025-021-01531-1}
\BIBentrySTDinterwordspacing

\bibitem{schniter2014compressive}
P.~Schniter and S.~Rangan, ``Compressive {P}hase {R}etrieval via {G}eneralized
  {A}pproximate {M}essage {P}assing,'' \emph{IEEE Transactions on Signal
  Processing}, vol.~63, no.~4, pp. 1043--1055, 2014.

\bibitem{SEGURA2011516}
\BIBentryALTinterwordspacing
J.~Segura, ``Bounds for ratios of modified {B}essel functions and associated
  {T}ur\'an-type inequalities,'' \emph{Journal of Mathematical Analysis and
  Applications}, vol. 374, no.~2, pp. 516--528, 2011. [Online]. Available:
  \url{https://www.sciencedirect.com/science/article/pii/S0022247X10007742}
\BIBentrySTDinterwordspacing

\bibitem{Natalini2010}
A.~Laforgia and P.~Natalini, ``Some {I}nequalities for {M}odified {B}essel
  {F}unctions,'' \emph{Journal of Inequalities and Applications}, no. 253035,
  2010.

\bibitem{I-mmse}
D.~Guo, S.~Shamai, and S.~Verd{\'u}, ``Mutual {I}nformation and {M}inimum
  {M}ean-{S}quare {E}rror in {G}aussian {C}hannels,'' \emph{IEEE Transactions
  on Information Theory}, vol.~51, no.~4, pp. 1261--1282, 2005.

\bibitem{csiszar1974extremum}
I.~Csisz{\'a}r, ``On an extremum problem of information theory,'' \emph{Studia
  Scientiarum Mathematicarum Hungarica}, vol.~9, no.~1, pp. 57--71, 1974.

\end{thebibliography}
 \bibliographystyle{IEEEtran}
\end{document}